\documentclass[12pt, a4paper]{extarticle}
\usepackage[margin=2cm]{geometry}
\usepackage[affil-it]{authblk}

\usepackage[english]{babel}
\usepackage{graphicx} 
\usepackage{wrapfig}
\usepackage{amsmath,amsthm,amssymb,scrextend}
\usepackage{booktabs}            
\usepackage{multirow}            
\usepackage{amsfonts}            
\usepackage{multicol}
\usepackage[shortlabels]{enumitem}
\usepackage{cutwin}
\usepackage{float}
\usepackage{csquotes}
\usepackage{enumitem}
\usepackage{tikz}
\usepackage{hyperref}

\def \rR {\mathbb{R}}
\def \bf#1 {\textbf{#1 }}
\def \sumt {\sum\limits}
\providecommand{\keywords}[1]
{
  \small	
  \textbf{\textit{Keywords:}} #1
}

\def \sumt {\sum\limits}
\DeclareMathOperator{\diam}{diam}
\DeclareMathOperator{\dist}{dist}

\DeclareMathOperator{\Rad}{Rad}

\renewenvironment{proof}{\begin{addmargin}[1em]{0em}\begin{newproof}}{\end{newproof}\end{addmargin}\qed}
\newtheorem{thm}{Theorem}
\newtheorem{lm}{Lemma}
\newtheorem{cor}{Corollary}

\begin{document}

\title{Relations between average clustering coefficient and another centralities in graphs}
\author{Mikhail Tuzhilin%
  \thanks{Affiliation: Moscow State University, Electronic address: \texttt{mtu93@mail.ru};}
}
\date{}
\maketitle

\begin{abstract}
Relations between average clustering coefficient and global clustering coefficient, local efficiency, radiality, closeness, betweenness and  stress centralities were obtained for simple graphs.
\end{abstract}



\section{Introduction.} 
The centrality measure was introduced by Bonacich in~\cite{Bonacich}. Centrality is a local (with relation to a vertex) or global (with relation to a whole graph) measures in networks. There are many centrality measures (or shortly centralities) such as local efficiency, radiality, closeness, betweenness, stress centralities, etc. Calculation of centralities is very useful for finding intrinsic properties of ``real'' networks (which can be found in applications)~\cite{Borgatti}-~\cite{Lee}. One of most important centrality measure is a clustering coefficient, that differentiate ``real'' graphs (or small-world networks) and random generated graphs~\cite{Watts}.

There are two definitions of clustering coefficient: the average clustering coefficient introduced by Watts-Strogatz~\cite{Watts} and the global clustering coefficient. It was shown in~\cite{Estrada} that for windmill graphs the average clustering coefficient and the global clustering coefficient asymptotically different. More precisely, the average clustering coefficient tends to 1 and the global clustering coefficient tends to 0 if the number of vertices increasing. In this paper, author provides two large class of graphs for which the average clustering coefficient is less or equal than the global cluster coefficient and vise versa.

Nowadays, there are also many articles where centrality measures are used for calculations and predictions of certain network characteristics, but a very few with theoretical basis. In the article~\cite{Strang} relations between different centralities were obtained, also an estimation of the local efficiency was obtained in terms of the average clustering coefficient. In this article relations between the average clustering coefficient and another centralities are proved for simple undirected graphs, in particular, it is proved that the estimation of local efficiency in terms of the average cluster coefficient is in fact an equality. 

\section{Main definitions.}

All subsequent definitions are given for a simple undirected graph $G$ without pendant vertices. It also can be defined to a simple graph with pendant vertices if every function where \{vertex degree $- 1$\} is in the denominator are defined to be equal to $0$ for all cases where vertices degrees equal to $1$, but this will be omitted in this article for the sake of brevity.

Let's give necessary denotations. Let's denote by
\begin{itemize}
    \item $V(G)$ the set of vertices,\;\; $E(G)$ the set of edges and\;\; $A = \{a_{ij}\}$ adjacency matrix of graph $G$.
    \item Neighbourhood $N(v)$ --- the set of vertices which adjacent to the vertex $v$,
    \item $N'(v) = N(v)\bigcup {v}$ subgraph in $G$ on these vertices,
    \item $\bar{f}(x_1, x_2, ... , x_k)$, where $f$ is any function $V\times V \times ... \times V\rightarrow\rR$, the restriction of this function on $N'(v)$ (for example $\bar L(x,y)$ will be the average shortest path between $x$ and $y$ with restriction to subgraph $N'(v)$),
    \item $d_i = deg(v_i)$,
    \item $n = |V(G)|,\;\; m = |E(G)|$,
    \item $X(i) = X(v_i)$ for any $X$ --- set or function corresponding to vertex $v_i$,
\end{itemize}

Let's give definitions of centralities:
\begin{enumerate}
    \item \bf{Diameter} $diam(G) = \max_{s,t\in V(G)} dist(s,t)$.
    \item \bf{Density} $D(G) = \frac {\text{number of edges in }G} {\text{maximum possible number of edges in }G}= \frac{2 m} {n(n-1)}$.
    \item \bf{Global efficiency} $E_{glob}(G) = \frac 1 {n (n-1)} \sumt_{s\neq t} \frac 1 {dist(s,t)}$.
    \item \bf{Average shortest path length} $L(G) = \frac 1 {n(n-1)} \sumt_{s\neq t} dist(s,t)$.
    \item \bf{Local cluster coefficient} 
    
    $c_i = c(i) = \frac {\text{number of edges in }N(i)} {\text{maximum possible number of edges in }N(i)}= \frac{2 |E(N(i))|} {d_i(d_i-1)}$. 
    \item \bf{Average clustering coefficient} 
    
    $C_{WS}(G) = \frac 1 {n} \sumt_{i\in V(G)} c_i = \frac 1 {n} \sumt_{i\in V(G)}  \frac{2 |E(N(i))|} {d_i(d_i-1)} = \frac 1 n \sumt_{i\in V(G)} \frac {\sumt_{j, k\in V(G)} a_{ij} a_{jk} a_{ki}} {d_i (d_i-1)}$.
    \item \bf{Global clustering coefficient} 
    
    $C(G) = \frac {\text{number of closed triplets in $G$}} {\text{number of all triplets in $G$}} = \frac {\sumt_{i, j, k\in V(G)} a_{ij} a_{jk} a_{ki}} {\sumt_{i\in V(G)} d_i (d_i-1)}$.
    \item \bf{Betweenness centrality} $BC(i) = \sumt_{s,t\in V(G),\; s\neq t\neq i} \frac {\sigma_{st}(i)} {\sigma_{st}}$, where $\sigma_{st}$ is the total number of shortest paths from $s$ to $t$ and $\sigma_{st}(i)$ is the total number of shortest paths which contains vertex $i$.
    \item \bf{Closeness centrality} $Clo(v) = \frac {n-1} {\sumt_{t\in V(G)} dist(v, t)}$. 
    \item \bf{Local efficiency} $E_{loc}(G) = \frac 1 {n} \sumt_{v\in V(G)} E_{glob}(N(v)).$
    \item \bf{Radiality} $Rad(v) = \frac {\sumt_{t\in V(G), t\neq v} (diam(G)+1-dist(v, t))} {n-1}$.
    \item \bf{Stress} $Str(i) = \sumt_{s,t\in V(G),\; s\neq t\neq i} \sigma_{st}(i)$, where $\sigma_{st}(i)$ is the total number of shortest paths from $s$ to $t$ which contains vertex $i$.
\end{enumerate}

Note that all centralities are non-negative and $D(G), E_{glob}, E_{loc}, c_i, C_{WS}, C(G)$ are less or equal 1.

\section{Main results.}
All subsequent lemmas and theorems are given for a simple undirected graph $G$ without pendant vertices. It also can be defined to a simple graph with pendant vertices if every function where $d_i-1$ is in the denominator are defined to be equal to $0$ for all cases $d_i = 1$.

Let's define in the similar way as for local efficiency \emph{the local average shortest path length, local betweenness, local stress and local radiality centralities}.

Let's 
\begin{enumerate}
    \item $L\bigl(N(i)\bigr) = \frac 1 {d_i(d_i-1)} \sumt_{v, w\in N(i)} \dist(v, w)$ be the average shortest path length for vertices of the neighborhood $N(i)$, where the average shortest path length is defined in the whole graph $G$,
    \item $BC(i, N(i)) = \sumt_{s,t\in N(i),\; s\neq t} \frac {\sigma_{st}(i)} {\sigma_{st}}$ be the betweenness centrality for vertices of the neighborhood $N(i)$, where shortest paths are defined in the whole graph $G$, and $BC_{loc}(G) = \frac 1 n \sumt_{i\in V(G)} \frac {BC(i, N(i))} {d_i(d_i-1)}$ --- \bf{the local betweenness centrality},
    \item $\Rad(v, N(i)) = \frac {\sumt_{t\in N(i), t\neq v} \big(\diam(N(i))+1-\dist(v, t)\big)} {d_i-1}$ be the  radiality centrality for vertices of the neighborhood $N(i)$, where shortest paths and the diameter are defined in the whole graph $G$, and $\Rad_{loc}(G) = \frac 1 n \sumt_{i\in V(G)} \frac {\sumt_{v\in N(i)} \Rad(v, N(i))} {d_i}$ --- \bf{the local radiality centrality}.
\end{enumerate}

Let's prove a lemma about a relation between average shortest path length between vertices in the neighborhood of $i$ and local clustering coefficient of this vertex.

\begin{lm}\label{lm1}
    $$L(N(i)) = 2-c_i.$$
\end{lm}
\begin{proof}
 $$
L(N(i)) = \frac 1 {d_i(d_i-1)} \sumt_{s,t\in N(i), s\neq t} dist(s,t) = \frac 1 {d_i(d_i-1)} \sumt_{(s,t)\in E(N(i))} dist(s,t)+ \sumt_{s,t\in N(i), (s,t)\notin E(N(i))} dist(s,t) = 
$$
$$
 = \frac 1 {d_i(d_i-1)} (2 |E(N(i))| + \sumt_{(s,i), (i, t)\in E(G), (s,t) \notin E(G)} dist(s,t)) = 
$$ 
$$ 
= \frac 1 {d_i(d_i-1)} (2 |E(N(i))|+2 (d_i(d_i-1) - 2 | E(N(i))|)) = 2 - c_i.
$$
Note that shortest paths for vertices in $N(i)$ are defined corresponding to whole graph $G$.
\end{proof}

Let's prove theorem about a connection between local efficiency and average clustering coefficient of a graph.
\begin{thm}
$$E_{loc}(G)= \frac 1 2 (1+C_{WS}(G)).$$
\end{thm}
\begin{proof}
Let's give two proofs of this fact:
\begin{enumerate}
\item  Note that by definition $D(N(i)) = c_i$. In the article~\cite{Strang} it was proved that
 $$
3-L(N(i))\leq 2 E_{glob}(N(i))\leq 1+D(N(i)).
$$
Using lemma~\ref{lm1}
$$
3-(2-c_i)\leq 2 E_{glob}(N(i))\leq 1+c_i.
$$
Note that shortest paths for vertices in $N(i)$ are defined corresponding to whole graph $G$. Averaging by $i$ ends the proof.
    \item Let's rewrite the local clustering coefficient formula:
$$
c_i = \frac {\sumt_{(s, t)\in E(N(i))} 1} {d_i(d_i-1)},
$$
$$
\frac 1 2 (1+c_i) = \frac 1 2  \frac {\sumt_{(s, t)\in E(N(i))} 1 + \sumt_{(s, t)\in E(N(i))} 1+ \sumt_{s,t\in V(N(i)), (s, t)\notin E(N(i))} 1} {d_i(d_i-1)} = 
$$
$$
 = \frac {\sumt_{(s, t)\in E(N(i))} 1 + \sumt_{s,t\in V(N(i)),(s, t)\notin E(N(i))} \frac 1 2} {d_i(d_i-1)} = \frac {\sumt_{s, t\in V(N(i))} \frac 1 {dist(s,t)} } {d_i(d_i-1)} = E_{glob}(N(i)).
$$
Averaging by $i$ ends the proof.
\end{enumerate}
\end{proof}

Let's prove theorem about a connection between average clustering coefficient and stress centrality.

\begin{thm}\label{thm2}
$$
C_{WS}(G)\geq 1- \frac 1 n \sumt_{i\in V(G)} \frac {Str(i)} {d_i(d_i-1)}.
$$
\end{thm}
\begin{proof}
Note that $\forall j,k\in N(i): (j,k)\notin E(N(i))$ the shortest path between $j$ and $k$ is $j\rightarrow i\rightarrow k$. Therefore,
$$
Str(i)\geq 2( \frac {d_i(d_i-1)} 2-|E(N(i))|),
$$
$$
\frac 1 {d_i(d_i-1)} Str(i)\geq 1-c_i,
$$
Averaging by $i$
$$
C_{WS}(G) \geq \frac 1 n \sumt_{i\in V(G)} (1- \frac {Str(i)} {d_i(d_i-1)}).
$$
Note that for $diam(G) = 2$ holds an equality.
\end{proof}

Let's prove theorem about a relation between average clustering coefficient and local betweenness centrality.

\begin{thm}\label{thm3}
$$
C_{WS}(G)\leq 1- BC_{loc}(G).
$$
\end{thm}
\begin{proof}
Let's note that
$$
BC(i, N(i)) = \sumt_{j,k\in N(i), \; (j,k)\notin E(N(i))} \frac {1} {\sigma_{jk}} \leq  \sumt_{j,k\in N(i), \; (j,k)\notin E(N(i))} 1 = d_i(d_i-1)- 2|E(N(i))|,
$$
$$
\frac {BC(i, N(i))} {d_i(d_i-1)}\leq 1-c_i.
$$
Averaging by $i$
$$
C_{WS}(G)\leq \frac 1 n \sumt_{i\in V(G)} (1- \frac {BC(i, N(i))} {d_i(d_i-1)}).
$$
Note that the equality holds, if there exists a unique shortest path between any two vertices of the neighbourhood $N(i)$. This holds if and only if there are no average shortest paths of the length 2 in $N(i)$, and thus $N(i)$ should be the disjoint union of several complete graphs for any vertex $i$. 
\end{proof}

By using theorems~\ref{thm2} and \ref{thm3} an estimation of average shortest path in the neighborhood of $i$ is obtained.
\begin{cor}
$$
\frac {BC(i, N(i))} {d_i(d_i-1)}\leq L(N(i))-1\leq \frac { Str(i)} {d_i(d_i-1)}.
$$
\end{cor}
Note that shortest paths for vertices in $N(i)$ are defined corresponding to whole graph $G$.

Let's prove lemma about a relation between average closeness centrality and average shortest path length in graph.

\begin{lm}\label{lm2}
$$
\frac 1 {n} \sumt_{v\in V(G)} Clo(v) \geq \frac {1} {L(G)}.
$$
\end{lm}
\begin{proof}
By the inequality of harmonic mean and arithmetic mean
$$
\frac 1 {n} \sumt_{v\in V(G)} Clo(v) = \frac 1 n \sumt_{v\in V(G)} \frac {n-1} {\sumt_{t\in V(G)} dist(v, t)}\geq \frac {n(n-1)} {\sumt_{v, t\in V(G)} dist(v, t)} = \frac 1 {L(G)}.
$$
Note that an equality holds when all average shortest path lengths from any vertex to all remaining vertices are equal.
\end{proof}

Now let's prove theorem about a relation between average clustering coefficient and closeness centrality.
\begin{thm}
$$
\frac 1 {2-C_{WS}(G)} \leq \frac 1 n \sumt_{i\in V(G)} \frac { \sumt_{v\in N(i)} \overline{Clo}(v)} {d_i}.
$$
\end{thm}

\begin{proof}
By lemma~\ref{lm2} 
$$
\frac 1 {d_i}  \sumt_{v\in N(i)} \overline{Clo}(v) \geq \frac 1 {L(N(i))} = \frac 1 {2-c_i}.
$$
By the inequality of harmonic mean and arithmetic mean (since $\forall i\in V(G), \;0\leq c_i\leq 1$):
$$
\frac 1 n \sumt_{i\in V(G)} \frac { \sumt_{v\in N(i)} \overline{Clo}(v)} {d_i} \geq \frac 1 n \sumt_{i\in V(G)}  \frac 1 {2-c_i}\geq \frac n {\sumt_{i\in V(G)} (2-c_i)} = \frac 1 {2-C_{WS}(G)}.
$$
\end{proof}

Let's prove lemma about a relation between average shortest path length and average radiality.

\begin{lm}\label{lm3}
$$
\frac 1 {n} \sumt_{v\in V(G)} {Rad}(v) = diam(G)+1 - L(G).
$$
\end{lm}
\begin{proof}
The proof holds from definition
$$
\frac 1 {n} \sumt_{v\in V(G)} {Rad}(v) = \frac 1 {n} \sumt_{v\in V(G)} \frac {(n-1) (diam(G)+1)- \sumt_{t\in V(G),\; t\neq v} dist(v,t))} {n-1} = diam(G)+1 -L(G).
$$
\end{proof}

Let's prove theorem about a relation between average clustering coefficient and local radiality.

\begin{thm}\label{thm:rad}
$$
C_{WS}(G) = \Rad_{loc}(G)-1 + \frac {\#\{N(i) \text{ which are complete graphs}\}} {n}.
$$
\end{thm}

\begin{proof}
By lemma~\ref{lm3}
$$
\frac 1 {d_i} \sumt_{v\in N(i)} \Rad(v, N(i)) =  \diam(N(i))+1-L(N(i)) = \diam (N(i))-1+c_i = c_i+1-\chi_{K_{d_i}}(N(i)),
$$
where 
$
\chi_{K_{d_i}}(N(i)) =
\begin{cases}
	1 & \text{if $N(i) = K_{d_i}$} \\
	0 & \text{otherwise}
\end{cases}
$. Averaging by $i$ ends the proof.
\end{proof}

Let's prove two theorems about a relation between average clustering coefficient and global clustering coefficient.
\begin{thm}\label{thm6} Let's $\forall i,j\in V(G)$ hold if $d_i\leq d_j $ then $ c_i\geq c_j$, then 
$$C_{WS}(G)\leq C(G).$$
\end{thm}
\begin{proof}
Let's re-numerate vertices such that $\forall i\leq j:d_i\leq d_j$.
Note that 
$$
c_i = \frac {\sumt_{j, k\in V(G)} a_{ij} a_{jk} a_{ki}} {d_i (d_i-1)},\;\; C(G) = \frac {\sumt_{i, j, k\in V(G)} a_{ij} a_{jk} a_{ki}} {\sumt_{i\in V(G)} d_i (d_i-1)}.
$$
Indeed,
$$ 
a_{ij} a_{jk} a_{ki} = 
\begin{cases}
	1 & \text{if there exists edge between vertices $j$ and $k$ which adjacent to vertex $i$} \\
	0 & \text{otherwise} 
\end{cases}
$$.
Therefore, 
$$
C_{WS}(G) = \frac 1 n \sumt_{i\in V(G)} \frac {\sumt_{j, k\in V(G)} a_{ij} a_{jk} a_{ki}} {d_i (d_i-1)}.
$$
Let's denote by $x_i = d_i (d_i-1)$. Since $|E(N(i))| = \frac 1 2 \sumt_{j, k\in V(G)} a_{ij} a_{jk} a_{ki}$ and the maximum number of edges in subgraph $N(i)$ equals to $\frac {d_i(d_i-1)} 2$, then $x_i\geq 2,\; 0\leq c_i\leq 1.$ Hence, using Chebyshev's sum inequality ($d_i\leq d_j \Rightarrow x_i\leq x_j \text{ and } c_i\leq c_j$):
$$
\frac 1 n \sumt_{i \in  V(G)} x_i\;  C_{WS}(G) = (\frac 1 n \sumt_{i \in  V(G)} x_i) (\frac 1 n \sumt_{i \in  V(G)} c_i) \leq \frac 1 n  \sumt_{i \in  V(G)} x_i c_i = \frac 1 n \sumt_{i, j, k\in V(G)} a_{ij} a_{jk} a_{ki}.
$$
Therefore,
$$
C_{WS}(G) \leq \frac {\sumt_{i, j, k\in V(G)} a_{ij} a_{jk} a_{ki}} {\sumt_{i \in  V(G)} d_i(d_i-1)}= C(G).
$$
The equality holds when $\forall i,j\in V(G): d_i = d_j$ (i.e. for regular graphs) or when $\forall i,j\in V(G): c_i = c_j$.
\end{proof}

\begin{cor}
Let's $\forall i,j\in V(G)$ hold if $d_i\leq d_j $, then $ c_i\geq c_j$, then 
$$C_{WS}(G)\geq C(G).$$
\end{cor}
The proof is the same as in theorem~\ref{thm6}.

\begin{cor}
For simple regular graphs $G$:
$$C_{WS}(G) = C(G).$$
\end{cor}

For most well-known graphs $C_{WS}(G)\geq C(G)$, but it is not very hard to come up with an example then $C_{WS}(G) < C(G).$ Consider the complete graph $K_n$ and glue to each vertex cycle of the length 4. Then, the vertices of the complete graph have by symmetry the same $c_i > 0$ and new added vertices have  $c_i = 0$. Thus, by theorem~\ref{thm6} for new graph $C_{WS}(G) < C(G).$

\end{document}